\documentclass[11pt]{article}

\usepackage[margin=1in]{geometry}
\usepackage{amsmath,amssymb,amsthm,mathtools,bm}
\usepackage{booktabs}
\usepackage{graphicx}
\usepackage{natbib}
\usepackage[hidelinks,breaklinks=true]{hyperref}
\usepackage[capitalize,noabbrev]{cleveref}

\hypersetup{
  pdftitle={Task-Oriented Quantization for Quadratic Scheduling: Centroid Water-Filling and Power-Diagram Encoders},
  pdfauthor={Joss Armstrong},
  pdfkeywords={task-oriented quantization, goal-oriented quantization, Lloyd-Max, water-filling, power diagram}
}

\newtheorem{proposition}{Proposition}

\theoremstyle{remark}
\newtheorem{remark}{Remark}

\title{Task-Oriented Quantization for Quadratic Scheduling:\\
Centroid Water-Filling and Power-Diagram Encoders}

\author{Joss Armstrong\\
Ericsson, Athlone, Ireland\\
\texttt{joss.armstrong@ericsson.com}}

\date{September 2026}

\begin{document}
\maketitle

\begin{abstract}
We distinguish two regimes in task-oriented quantization with a known
deterministic oracle action.  For an unconstrained interior oracle and a
smooth strongly concave utility, quantizing the oracle action by vector
Lloyd--Max minimizes a mean-squared-error surrogate and achieves a
$\beta/\alpha$ approximation to the optimal $K$-level task quantizer.  The
reduction is exact for isotropic quadratic loss, and the corresponding
task rate--distortion function is bracketed by two ordinary
rate--distortion functions.  Budget-constrained quadratic scheduling is
different: the oracle satisfies a variational inequality, so quantizing
water-filled actions is not generally optimal.  We derive the exact
Lloyd-type conditions for this case.  The optimal action for a quantizer
cell is water-filling evaluated at the cell's conditional-mean load, and
the optimal encoder partitions load space into affine power-diagram cells.
Thus the correct prescription is to quantize the load and water-fill its
centroid.  The distinction is material whenever a cell crosses
water-filling active-set boundaries.
\end{abstract}


%% ====================================================================
\section{Introduction}
%% ====================================================================

Task-oriented communications~\cite{gunduz2023} replaces
reconstruction-fidelity criteria with criteria tied to the receiver's
downstream decision. In wireless task-oriented quantization, this has
produced bespoke quantizer designs whose encoder and decoder rules are
driven by the task loss rather than by Euclidean distortion.
\citet{zou2023} introduce goal-oriented quantization (GOQ) with
high-resolution optimality conditions involving the Jacobian of the
decision rule and the Hessian of the goal (Propositions~1--2,
Algorithm~1). \citet{sun2024} extend GOQ to joint precoding
and quantization for $L_p$-norm power scheduling, with a modified
Lloyd iteration whose encoder and decoder conditions are driven by the
task loss (Algorithm~2, Eqs.~26--28). Adjacent threads include
task-based quantization for hardware-limited
inference~\cite{shlezinger2019} and the wider semantic-communication
literature~\cite{gunduz2023}.

A common structure recurs in wireless-resource testbeds: the receiver
applies a known optimal action to the decoded representation rather than
inferring a latent quantity.  This structure alone, however, does not make
Euclidean quantization of the oracle action optimal.  The geometry depends
on whether the oracle is an interior stationary point or a constrained
optimum.  In the former case, a curvature sandwich controls task loss by
mean-squared error (MSE) on the oracle action.  In the latter, the
first-order term is governed by a variational inequality and need not
vanish when a reconstruction crosses a constraint boundary.

This paper makes that distinction explicit.  First, for unconstrained
interior oracles, we prove an MSE sandwich, correct rate--distortion bounds,
and a finite-codebook $\beta/\alpha$ guarantee for Lloyd--Max on the oracle
action.  Exact equivalence is reserved for isotropic quadratic loss.
Second, for the $L_2$ power-scheduling problem of \citet{sun2024}, we
solve the constrained cell update exactly.  A cell's
representative is not, in general, the conditional mean of its
water-filled schedules; it is the water-filled conditional-mean load.  For
fixed representatives, the assignment rule is affine in the load and
therefore produces a power diagram.  These two updates give a monotone
Lloyd-type algorithm with the correct constrained geometry.

%% ====================================================================
\section{Designed-Source Setup}
%% ====================================================================

Source $T \sim F$ over $\mathcal{T} \subseteq \mathbb{R}^d$. The receiver applies a
known deterministic optimal action
$\phi^*: \mathcal{T} \to \mathbb{R}^r$ to the decoded representation,
incurring task loss
$\ell(t, a) = U(t, \phi^*(t)) - U(t, a) \geq 0$ where
$U: \mathcal{T} \times \mathcal{A} \to \mathbb{R}$ is the utility being
maximised. We refer to such sources as \emph{designed} because the
receiver's optimal action is fixed by the system design, in contrast
to settings where the task-relevant quantity is a latent variable to
be inferred. Let
\begin{equation}\label{eq:within-cell}
\varepsilon(c) =
\mathbb{E}\bigl[\|\phi^*(T) - \phi(c(T))\|^2\bigr],
\quad \phi(k) = \mathbb{E}[\phi^*(T)\mid c(T){=}k],
\end{equation}
denote the within-cell variance for any quantizer $c: \mathcal{T} \to
\{1,\ldots,K\}$ with conditional-mean decoder $\phi$.

For the results in this section we assume:
\begin{itemize}
\item[\textbf{A0.}] For every $t$, $\phi^*(t)$ is an unconstrained
  maximizer of the differentiable map $a\mapsto U(t,a)$ on
  $\mathbb{R}^r$; hence $\nabla_aU(t,\phi^*(t))=0$.
\item[\textbf{A1.}] $\beta$-smoothness of $U$ in the action:
  $\nabla^2_{aa} U(t, a) \succeq -\beta I$ for all $(t, a)$.
\item[\textbf{A2.}] $\alpha$-strong concavity of $U$ in the action:
  $\nabla^2_{aa} U(t, a) \preceq -\alpha I$ for all $(t, a)$.
\end{itemize}
A1 and A2 are standard $\beta$-smoothness and $\alpha$-strong-concavity
conditions from convex analysis~\cite{boyd2004}.  Together with the
oracle action $\phi^*$, they define the designed-source mechanism
class of \cite{mises_preprint}.  A0 is load-bearing.
If actions are constrained to a budget simplex, the oracle generally
satisfies a variational inequality rather than a zero-gradient condition;
that case is analysed separately in \Cref{sec:wireless}.

For comparison, the GOQ framework~\cite{zou2023} considers an
$M$-level quantizer $Q_M: \mathcal{T} \to \{z_1, \ldots, z_M\}$ with
optimality loss
\begin{equation}\label{eq:ol}
L(Q; f) = \alpha_f \int_{\mathcal{T}}
  \bigl[f(\chi(Q(g));\, g) - f(\chi(g);\, g)\bigr]\,\varphi(g)\,dg,
\end{equation}
where $\chi(g) = \arg\min_x f(x; g)$, $f$ is the goal function, and
$\varphi$ is the density of $g$. Identifying $g \leftrightarrow t$,
$\chi \leftrightarrow \phi^*$, $f(\chi(g);g) \leftrightarrow -U(t, a)$,
\eqref{eq:ol} matches the designed-source task loss above. The
task-oriented rate-distortion function is
\begin{equation}\label{eq:rtask}
R_{\mathrm{task}}(D) = \min_{\substack{p(\hat{a} \mid t):\\
  \mathbb{E}[\ell(t, \hat A)] \leq D}} I(T;\, \hat A),
\end{equation}
where $\hat A$ is the reproduction action.

%% ====================================================================
\section{Main Results}
%% ====================================================================

\begin{proposition}[Curvature sandwich]\label{prop:task-loss}
Under A0--A2, every reproduction action $\hat A$ satisfies
\begin{equation}\label{eq:sandwich}
\frac{\alpha}{2}\,
\mathbb{E}\|\phi^*(T)-\hat A\|^2
\;\leq\; \mathbb{E}\ell(T,\hat A)
\;\leq\; \frac{\beta}{2}\,
\mathbb{E}\|\phi^*(T)-\hat A\|^2.
\end{equation}
If $U(t,a)=u(t)-(\gamma/2)\|a-\phi^*(t)\|^2$, both inequalities
are equalities with $\alpha=\beta=\gamma$.
\end{proposition}

\begin{proof}
Expand $U(t,\cdot)$ around $\phi^*(t)$.  The linear term vanishes by
A0.  The Hessian bounds in A1--A2 give the two pointwise quadratic
bounds; taking expectations proves \eqref{eq:sandwich}.  The final
assertion follows by direct substitution.
\end{proof}

\begin{proposition}[Rate--distortion sandwich]\label{prop:rd}
Under A0--A2,
\begin{equation}\label{eq:rd-sandwich}
R_{\mathrm{MSE}}^{\phi^*}\!\!\left(\frac{2D}{\alpha}\right)
\;\leq\; R_{\mathrm{task}}(D)
\;\leq\; R_{\mathrm{MSE}}^{\phi^*}\!\!\left(\frac{2D}{\beta}\right)
\end{equation}
where $R_{\mathrm{MSE}}^{\phi^*}(\delta)$ is the standard MSE
rate-distortion function for the source $\phi^*(T)$. For the isotropic
quadratic utility in \Cref{prop:task-loss},
$R_{\mathrm{task}}(D)=R_{\mathrm{MSE}}^{\phi^*}(2D/\gamma)$.
\end{proposition}

\begin{proof}
Every task-feasible test channel has MSE at most $2D/\alpha$, giving
the lower bound.  Conversely, every test channel with MSE at most
$2D/\beta$ is task-feasible, giving the upper bound.  Pointwise
equality of the two distortions gives the last assertion.
\end{proof}

\begin{proposition}[Lloyd--Max surrogate guarantee]\label{prop:quantizer}
Under A0--A2, let $Q_{\rm LM}$ be a globally MSE-optimal $K$-level
quantizer of $\phi^*(T)$, and let $D_K^*$ be the minimum task distortion
over all $K$-level encoders and decoders.  Then
\begin{equation}\label{eq:approximation}
D(Q_{\rm LM})\leq \frac{\beta}{\alpha}D_K^*.
\end{equation}
For isotropic quadratic loss, $Q_{\rm LM}$ is task-optimal.
\end{proposition}

\begin{proof}
For any task-optimal quantizer $Q^*$, MSE optimality and
\Cref{prop:task-loss} give
\[
D(Q_{\rm LM})\leq \frac{\beta}{2}M(Q_{\rm LM})
\leq \frac{\beta}{2}M(Q^*)
\leq \frac{\beta}{\alpha}D(Q^*).
\]
Here $M(Q):=\mathbb E\|\phi^*(T)-\hat\phi_Q(C)\|^2$.
The conditional-mean representatives used by $Q_{\rm LM}$ are MSE
optimal~\cite{lloyd1982,gersho1992}.  In the isotropic quadratic case,
task loss is a constant multiple of MSE, proving exact optimality.
\end{proof}

\begin{remark}[Reduction of GOQ to Lloyd-Max]\label{rem:goq}
\citet{zou2023} derive the GOQ algorithm (Algorithm~1) with
a task-weighted nearest-neighbour condition using the matrix
$E_{f,\chi}(g) = B_{f,\chi}(g) + A_{f,\chi}(g)$, where
$A_{f,\chi} = J_\chi^\top H_f J_\chi$ (Proposition~2), and a
gradient-based representative update. They note (p.~48) that when the
first-order optimality condition holds, $B_{f,\chi} = 0$.
\citet{sun2024} state the conditions most cleanly, with the
encoder assigning by task loss (Eq.~26) and the decoder minimizing task
loss (Eq.~28).

Under A0--A2, Euclidean Lloyd--Max on $\phi^*(t)$ is a controlled
surrogate, not generally an exact rewriting of the task-weighted
conditions.  The exact reduction follows when task loss is isotropic
quadratic in $a-\phi^*(t)$.  Constraint activity also prevents the
ordinary first-order cancellation, as the water-filling example below
shows.
\end{remark}

\begin{proposition}[Gaussian linear oracle]\label{prop:dim}
Let $T \sim \mathcal{N}(0,\sigma^2 I_d)$, let
$\phi^*(t)=At$ with $AA^\top=I_r$, and let task loss be
$\|At-a\|^2$.  At common per-coordinate distortion
$0<D<\sigma^2$, the difference between reconstructing all $d$ source
coordinates and reconstructing the $r$ task coordinates is
\begin{equation}\label{eq:rate-saving}
R_{\mathrm{recon}}(D) - R_{\mathrm{task}}(D)
\;=\; \frac{d - r}{2}\,\log\frac{\sigma^2}{D}\quad\text{nats/symbol}.
\end{equation}
\end{proposition}

\begin{proof}
Since $AT\sim\mathcal{N}(0,\sigma^2I_r)$, the isotropic quadratic
case of \Cref{prop:rd} reduces task coding to ordinary Gaussian
rate--distortion coding of $AT$.  Apply
$R_n(D)=(n/2)\log(\sigma^2/D)$~\cite{cover2006} at dimensions $d$
and $r$ and subtract.
\end{proof}

%% ====================================================================
\section{Examples}
%% ====================================================================

\subsection{Gaussian Type, Linear Oracle}

Let $T \sim \mathcal{N}(0, \Sigma_T)$ with
$\Sigma_T \in \mathbb{R}^{d \times d}$, and
$\phi^*(t) = At$ where $A \in \mathbb{R}^{r \times d}$, $r < d$.
Utility: $U(t, a) = -(a - At)^\top(a - At)$, so
$\alpha = \beta = 2$ (quadratic).

\smallskip\noindent
\textbf{Standard derivation.} Characterize the distortion-rate function
for the task loss $\ell(t, a) = \|At - a\|^2$, optimizing over
encoders and decoders with a non-standard distortion measure. The
solution requires eigendecomposition of $A\Sigma_T A^\top$ and reverse
water-filling in the eigenspace.

\smallskip\noindent
\textbf{Reduction.}
\begin{enumerate}
\item $\phi^*(T) = AT \sim \mathcal{N}(0, A\Sigma_T A^\top)$, an
  $r$-dimensional Gaussian.
\item By \Cref{prop:rd} (tight for quadratic):
  $R_{\mathrm{task}}(D) = R_{\mathrm{MSE}}^{\phi^*}(D)$ = reverse
  water-filling on the eigenvalues of $A\Sigma_T A^\top$.
\item The $K$-level task-quantization objective is exactly the standard
  vector-quantization objective on $AT$ (\Cref{prop:quantizer}).
\end{enumerate}
For the white, orthonormal specialization of \Cref{prop:dim}, the rate
saving is $((d-r)/2)\log(\sigma^2/D)$ nats per symbol.

\begin{figure}[t]
\centering
\includegraphics[width=\columnwidth]{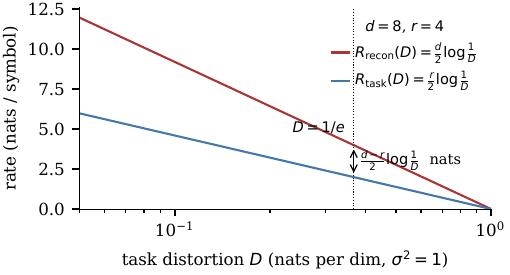}
\caption{Rate-distortion curves for white Gaussian
$T\sim\mathcal{N}(0,I)$, $d=8$, with linear oracle $\phi^*(t)=At$ to
an $r=4$ dimensional action. The vertical gap is the rate saving from
discarding the $d-r$ task-irrelevant dimensions.}
\label{fig:rd}
\end{figure}

\subsection{Power Scheduling with Water-Filling Oracle}\label{sec:wireless}

Consider the $p=2$ power-scheduling testbed of \citet{sun2024}.
A scheduler observes baseline loads
$t = (t_1, \ldots, t_N) \in \mathbb{R}_{\geq 0}^N$ over $N$ time slots
and allocates a controllable demand $a \in \mathbb{R}_{\geq 0}^N$
in the simplex
$\mathcal A_E:=\{a\geq0:\bm 1^\top a=E\}$.  The smoothing objective
$f(t,a)=\|t+a\|^2$ is minimised at
\begin{equation}\label{eq:waterfill}
\phi^*(t)=\bigl(\lambda(t)\bm 1-t\bigr)^+,
\end{equation}
where $\lambda(t)$ makes the coordinates sum to $E$.  The induced task
loss is $\ell(t,a)=f(t,a)-f(t,\phi^*(t))$.  Although its Hessian in $a$
is $2I$, A0 fails: the gradient at $\phi^*(t)$ is nonzero whenever the
simplex constraint is active, as it always is here.

\begin{proposition}[Exact constrained Lloyd conditions]
\label{prop:constrained-lloyd}
For any measurable $K$-cell partition $\{\mathcal C_k\}$ with
$\Pr(T\in\mathcal C_k)>0$, define
$\mu_k=\mathbb E[T\mid T\in\mathcal C_k]$.  The task-optimal action in
cell $k$ is
\begin{equation}\label{eq:centroid-wf}
a_k^*=\operatorname{WF}_E(\mu_k)
=\bigl(\lambda_k\bm1-\mu_k\bigr)^+.
\end{equation}
Conversely, for fixed feasible actions $a_1,\ldots,a_K$, the optimal
encoder assigns $t$ to any minimizer of
\begin{equation}\label{eq:power-assignment}
2t^\top a_k+\|a_k\|^2.
\end{equation}
Consequently, each pairwise cell boundary is a hyperplane and the cells
form an affine power diagram in load space.
\end{proposition}

\begin{proof}
The oracle term in $\ell(t,a)$ is independent of the quantizer.  For a
fixed cell,
\[
\mathbb E[f(T,a)\mid\mathcal C_k]
=\mathbb E\|T\|^2+2\mu_k^\top a+\|a\|^2.
\]
Minimizing the last two terms over $\mathcal A_E$ is the Euclidean
projection of $-\mu_k$ onto the simplex, whose KKT solution is
\eqref{eq:centroid-wf}.  For fixed actions, removing the common term
$\|t\|^2$ from $f(t,a_k)$ gives \eqref{eq:power-assignment}.  Comparing
indices $j$ and $k$ gives an affine half-space.
\end{proof}

Alternating \eqref{eq:centroid-wf} and
\eqref{eq:power-assignment} cannot increase expected task distortion,
because each update exactly minimizes it with the other block fixed.
As with ordinary Lloyd iteration, this establishes monotone descent to a
coordinatewise optimum, not global optimality from arbitrary
initialization.

The distinction from oracle-action quantization is exposed by the exact
identity
\begin{equation}\label{eq:constrained-identity}
\ell(t,a)=\|a-\phi^*(t)\|^2
+2\!\sum_{n:\,t_n>\lambda(t)}
  \bigl(t_n-\lambda(t)\bigr)a_n.
\end{equation}
The second term is nonnegative and is strictly positive when the decoded
action allocates energy to a slot inactive for the true load.  It vanishes
for cells contained within one active-set region; there water-filling is
affine and $\operatorname{WF}_E(\mathbb E[T\mid\mathcal C_k])
=\mathbb E[\phi^*(T)\mid\mathcal C_k]$.  Across active-set boundaries,
the two representatives generally differ.  For example, with $N=2$,
$E=1$, $t=(0,10)$, and $a=(1/2,1/2)$, the task loss is $9.5$ whereas
$\|a-\phi^*(t)\|^2=0.5$.

\begin{figure}[t]
\centering
\includegraphics[width=\columnwidth]{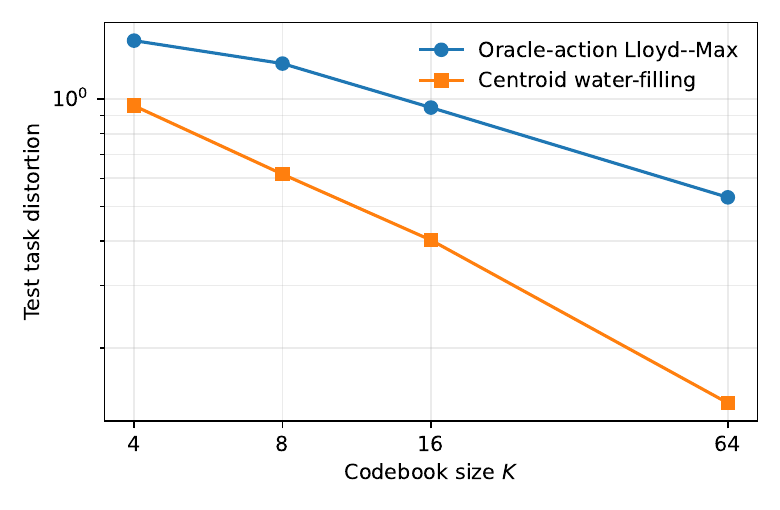}
\caption{Held-out task distortion for $N=8$, $E=1.6$, and independent
unit-mean exponential loads.  Each point is the best of five matched
initializations, trained on $5\times10^4$ samples and evaluated on
$10^5$ fresh samples.  Oracle-action Lloyd--Max incurs respectively
$1.52$, $2.05$, $2.36$, and $3.78$ times the distortion of centroid
water-filling at $K=4,8,16,64$.}
\label{fig:comparison}
\end{figure}

\Cref{fig:comparison} compares the two alternating designs.  The same
water-filling implementation, samples, and k-means++ initial actions are
used for both methods in each restart.  The widening finite-$K$ gap shows
that MSE on the oracle action is not a task-distortion proxy in this
constrained testbed.

%% ====================================================================
\section{Discussion}
%% ====================================================================

The zero-gradient condition, rather than determinism of the oracle alone,
is what permits an oracle-space MSE reduction.  Under A0--A2,
Lloyd--Max supplies a $\beta/\alpha$ approximation and the
rate--distortion sandwich in \Cref{prop:rd}; exact equivalence requires
the additional isotropic-quadratic structure.  The Gaussian linear-oracle
example has that structure and remains an exact standard
rate--distortion problem.

Budget-constrained scheduling falls outside A0.  Its correct simplification
is nevertheless elementary and operational: compute centroids in load
space, water-fill those centroids, and use the resulting actions to form
power-diagram encoder cells.  This keeps the standard alternating-design
workflow while respecting the simplex KKT conditions.  At high resolution,
cells that remain inside one active-set region recover oracle-action
Lloyd--Max locally; \eqref{eq:constrained-identity} quantifies why that
approximation can fail at finite resolution.

The present analysis concerns deterministic designed actions.  General
goal-oriented quantization remains necessary for latent stochastic tasks
and for losses whose cell updates have no closed form
\cite{zou2023,gunduz2023}.  Even in the quadratic scheduling case, the
alternating algorithm is nonconvex jointly in cells and actions, so
initialization and empirical comparison remain important.

\bibliographystyle{plainnat}
\bibliography{refs}

\end{document}